\documentclass[11pt]{article}
\usepackage{fullpage}
\usepackage{amssymb}
\usepackage{times}
\usepackage{amsmath,amsthm}
\usepackage{enumitem}
\usepackage{graphicx}
\usepackage{xcolor}
\usepackage{url}
\usepackage{cite}
\usepackage{microtype}
\usepackage{hyperref}
\usepackage{mathtools}
\usepackage[linesnumbered,ruled,vlined]{algorithm2e}
\usepackage[framemethod=tikz]{mdframed}

\newtheorem{theorem}{Theorem}[section]

\newtheorem{lemma}[theorem]{Lemma}
\newtheorem{corollary}[theorem]{Corollary}

\newcommand{\epl} {\epsilon_{\scriptscriptstyle\mathsf{L}}}
\newcommand{\epu} {\epsilon_{\scriptscriptstyle\mathsf{U}}}
\newcommand{\epub} {\bar{\epsilon}_{\scriptscriptstyle\mathsf{U}}}

\title{Local Dominance in Mixed-Strength Populations \\ Fast Maximal Independent Set}

\author{
Michael Luby\thanks{BitRipple}
\and
Sandy Irani\thanks{Department of Computer Science, University of California, Irvine}
}

\date{} 

\begin{document}
\maketitle

\begin{abstract}
In many natural and engineered systems, agents interact through repeated local contests that determine which individuals become dominant within their neighborhoods. These interactions are shaped by inherent differences in strength, and they often lead to stable dominance patterns that emerge surprisingly quickly relative to the size of the population. This motivates the search for simple mathematical models that capture both heterogeneous agent strength and rapid convergence to stable local dominance.

A widely studied abstraction of local dominance is the Maximal Independent Set (MIS) problem. In the classical MIS protocol of~\cite{luby1986}, each agent repeatedly generates a strength value chosen uniformly and becomes locally dominant if its value is smaller than those of its neighbors, and~\cite{luby1986} proves that the protocol converges quickly. Since all agents draw from the same uniform distribution, this provides a theoretical explanation for fast dominance convergence in populations of equal-strength agents. This naturally raises the question of whether fast convergence also holds in the more realistic setting where agents are inherently mixed-strength.

To investigate this question, we formalize the \emph{mixed-strength agents} model, in which each agent draws its strength from its own distribution, and these per-agent distributions may vary across the population. We prove that the extension of the~\cite{luby1986} protocol where each agent repeatedly generates a strength value from its own distribution still exhibits fast dominance convergence, providing formal confirmation of the rapid convergence observed in many mixed-strength natural processes.

We also show that heterogeneity can significantly change the dynamics of the process. In contrast to the equal-strength setting, a constant fraction of edges need not be eliminated per round. We construct a population and strength profile in which progress per round is asymptotically smaller, illustrating how inherent strength asymmetry produces qualitatively different global behavior.
\end{abstract}

\section{Introduction}

In many natural and engineered systems, agents interact through repeated local contests that determine which individuals become dominant within their neighborhoods. These interactions are shaped by inherent differences in strength, and they often lead to stable dominance patterns that emerge surprisingly quickly relative to the size of the population. This motivates the study of simple mathematical models that capture both heterogeneous agent strength and rapid convergence to stable local dominance.

A widely studied abstraction of local dominance is the Maximal Independent Set (MIS) problem. In this view, each agent competes with its neighbors, and an agent becomes locally dominant if it is chosen as a member of a MIS. The classical MIS protocol of~\cite{luby1986} provides a particularly simple model of this process. In that protocol, agents repeatedly generate random strength values, compare them to those of their neighbors, and join the MIS if their value is the smallest in their local neighborhood. All agents draw from the same uniform distribution, corresponding to a population of inherently equal-strength agents. The analysis in~\cite{luby1986} shows that the protocol eliminates a constant fraction of remaining edges on average in each round and thus converges in $O(\log(n))$ rounds with high probability. This gives a theoretical explanation for fast dominance convergence in symmetric populations.

Many real-world systems, however, consist of heterogeneous agents. Biological, social, and engineered populations often include individuals with differing levels of energy, capability, or environmental advantage. Local contests among such mixed-strength agents produce dominance patterns that stabilize quickly, even though the agents are not symmetric. This raises a natural question: does the classical MIS protocol retain fast convergence when agents have inherent differences in strength?

To address this question, we study the \emph{mixed-strength agents} model, in which each agent draws its strength from its own distribution over $[0,1]$. These per-agent distributions may vary widely across the population, capturing heterogeneous competitive ability while preserving the local structure and simplicity of the classical MIS protocol. In this model, an agent with a distribution skewed toward smaller values can be viewed as inherently stronger than its neighbors.

Our main result shows that the classical MIS protocol of~\cite{luby1986} continues to converge quickly in mixed-strength populations. Under mild technical conditions on the agent-specific cumulative distribution functions, we prove that the process terminates in
\[
O(\log(n) \cdot \log(d))
\]
rounds with high probability, where $n$ is the population size and $d$ is the maximum degree of the graph. This demonstrates that simple decentralized competition yields rapid dominance convergence even under heterogeneity, providing formal confirmation of a phenomenon widely observed in natural systems.

We also show that heterogeneity can significantly change the dynamics of the process. In contrast to the equal-strength setting, a constant fraction of edges need not be eliminated per round. In Section~\ref{sec:few edges}, we construct a graph and mixed-strength profile in which progress per round is asymptotically smaller, illustrating how asymmetric strength can produce qualitatively different global behavior.

The remainder of the paper is organized as follows. Section~\ref{sec:mis} reviews the classical MIS protocol in the equal-strength setting. Section~\ref{sec:related-work} discusses related work on distributed MIS selection and symmetry breaking. Section~\ref{sec:model} introduces the mixed-strength agents model and presents natural examples. Section~\ref{analysis} proves fast convergence in mixed-strength populations. Section~\ref{sec:few edges} gives an example where asymmetry slows local elimination. Section~\ref{sec:discussion} concludes with broader implications and open questions.

\section{The MIS protocol}
\label{sec:mis}

The classic MIS protocol of~\cite{luby1986} can be summarized as follows:

\begin{center}
\begin{minipage}{0.95\linewidth}
\begin{mdframed}[linewidth=0.5pt, roundcorner=4pt, innertopmargin=0.8em, innerbottommargin=0.8em, innerleftmargin=1em, innerrightmargin=1em, backgroundcolor=gray!5]

\textbf{The MIS Protocol}
\vspace{0.5em}

\noindent
All agents begin in the \textit{active} state.

\vspace{0.5em}
\noindent
The protocol proceeds in synchronous rounds. In each round:

\begin{itemize}
  \item Each active agent independently selects a random number uniformly from the interval $[0,1]$.
  \item An active agent that has the largest number among its active neighbors adds itself to the MIS and becomes \textit{inactive} (chosen as a leader).
  \item Any active agent adjacent to a newly added MIS member becomes \textit{inactive} (it is adjacent to a leader).
\end{itemize}

\end{mdframed}
\end{minipage}
\end{center}

This protocol assumes that all agents are equal in strength, meaning they execute the same logic, follow identical rules, and sample random values from the same uniform distribution. No agent has any structural advantage or distinguishing feature that influences its probability of selection.

Under these assumptions, it can be shown that in each round, a constant fraction of the edges between active agents is eliminated on average. As a result, the protocol terminates when no active agents remain with high probability in $O(\log n)$ rounds, where $n$ is the number of agents in the graph. This setting of uniform behavior and symmetry has been central to much of the theoretical work in distributed symmetry breaking~\cite{alon91, ghaffari16}, where agents are treated as interchangeable and coordination is driven entirely by probabilistic fairness.

Our main result shows that the classical MIS protocol of~\cite{luby1986} retains fast convergence even under agent asymmetry. Specifically, we prove in Theorem~\ref{thm:main} that the protocol completes with high probability in
\[
O(\log(n) \cdot \log(d))
\]
rounds, where $n$ is the number of agents and $d$ is the maximum degree in the graph, provided each agent samples its random number in each round from a cumulative distribution function (CDF) satisfying the Mixed-Strength Agents Conditions defined in Section~\ref{sec:model}.

\section{Related Work}
\label{sec:related-work}

The classical MIS protocol of~\cite{luby1986} assumes agents draw their values from a common uniform distribution. This symmetry has underpinned much of the theoretical work in distributed symmetry breaking~\cite{alon91, ghaffari16}, where agents are treated as interchangeable and randomized coordination is analyzed under homogeneous assumptions.

While the {\bf Beeping model}~\cite{alon2011} and its successors~\cite{EW13, EU20, EK21} explore extreme communication constraints (such as lack of unique identifiers, limited sensing, or one-bit messages), they still rely on globally uniform random choices by agents. These works emphasize simplicity and local information but do not model population-level heterogeneity in agent behavior.

A recent line of work on asynchronous and dynamic MIS protocols, such as~\cite{emek2023asynchronous}, begins to relax assumptions about network stability or timing by allowing for random delays or topological changes. However, even these protocols typically assume agents use the same distribution when generating their random decisions.

In contrast, our model departs from this tradition by allowing each agent to draw from an individual, potentially distinct distribution (capturing real-world asymmetries in capability, environment, or state). To the best of our knowledge, no prior formal analysis has studied MIS under such heterogeneous randomness. We show that even with this strong asymmetry, the classic protocol converges in polylogarithmic time. This opens the door to studying population heterogeneity in a range of distributed symmetry-breaking problems.

\section{Mixed-Strength Agents Model}
\label{sec:model}

We introduce the \textbf{mixed-strength agents} model to capture systems in which agents differ in their inherent ability to win local competition. Unlike the standard MIS setting, where all agents sample uniformly at random from $[0,1]$ and are thus symmetric in expectation, our model allows each agent to sample from its own individual distribution. This reflects asymmetries in capability that arise in both natural and engineered systems, such as differences in influence, resources, or connectivity.

In this setting, each agent may use a distinct probability distribution over $[0,1]$, making it more or less likely to generate lower values. Agents whose distributions place more weight on smaller values have a competitive advantage in the MIS selection process and are more likely to become local leaders. These asymmetries affect the dynamics of the protocol and allow us to model heterogeneous populations more realistically.

To formalize this, we associate each agent $i$ with a cumulative distribution function (CDF) $p_i[x]$ over $[0,1]$, where
\[
p_i[x] = \Pr\left[\text{number chosen by agent } i \text{ in a round} \le x\right].
\]
This CDF governs how agent $i$ samples its random number in each round. Strength is encoded in the shape of the distribution: agents with CDFs skewed toward smaller values are inherently stronger in local competition. By analyzing these per-agent CDFs, we can study how heterogeneity influences MIS formation and quantify the impact of strength asymmetries on the protocol’s behavior.

We now describe two natural examples that instantiate this model:

\begin{description}
    \item[Fair bits:] Each bit of the number is independently set to 0 or 1 with equal probability $1/2$. This yields a distribution that is uniform over $[0,1]$, and in this case, the CDF for agent $i$ is $p_i[x] = x$ for all $x$. This corresponds to the classical, symmetric setting.

    \item[Biased bits:] Fix constants $\epl$ and $\epu$ such that $0 < \epl \le \epu < 1$. Each agent $i$ is assigned a fixed bit bias parameter $q_i \in [\epl, \epu]$. To generate a number in $[0,1]$, each bit is chosen independently to be 0 with probability $q_i$ and 1 with probability $1 - q_i$. The resulting number is more likely to fall closer to 0 when $q_i > 1/2$, and closer to 1 when $q_i < 1/2$. The CDF for agent $i$ satisfies:
    \[
    p_i\left[\frac{1}{2^\ell}\right] = q_i^\ell
    \]
    for all nonnegative integers $\ell$.
\end{description}

In the biased bits model, the bit bias $q_i$ reflects the inherent strength of agent $i$: a larger $q_i$ corresponds to a stronger agent, more likely to generate a smaller number and win in local competition. To illustrate the effect of this asymmetry, consider two classes of agents that use bits with biases in the interval $\left[\frac{1}{4},\frac{1}{2}\right]$:

\begin{description}
    \item[Stronger agents:] $q_i = 1/2$.
    \item[Weaker agents:] $q_i = 1/4$.
\end{description}

Suppose agent $i$ has $d$ stronger neighbors. Then, in a round of the protocol:
\begin{itemize}
    \item If $i$ is a stronger agent, it is added to the MIS with probability approximately $\frac{1}{d}$.
    \item If $i$ is a weaker agent, it is added to the MIS with probability approximately $\frac{1}{d^2}$.
\end{itemize}

Suppose agent $i$ has $d$ weaker neighbors. Then, in a round of the protocol:
\begin{itemize}
    \item If $i$ is a stronger agent, it is added to the MIS with probability approximately $\frac{1}{\sqrt{d}}$.
    \item If $i$ is a weaker agent, it is added to the MIS with probability approximately $\frac{1}{d}$.
\end{itemize}

These examples demonstrate how an agent’s inherent strength, captured by its distribution bias, can dramatically influence its chance of becoming a local leader.

To analyze the protocol's behavior under heterogeneous CDFs, we introduce two technical conditions that constrain the family of allowable CDFs on the random numbers chosen by agents. These conditions are satisfied by the examples above, and by a broad class of other natural distributions.

\vspace{0.2cm}
\noindent
{\bf Mixed-Strength Agents Conditions:} Fix constants $\epl$ and $\epu$ such that $0 < \epl \le \epu < 1$. We require that for all agents $i$ and for all nonnegative integers $\ell$, the CDF $p_i[x]$ for agent $i$ satisfies:
\begin{align}
p_i\left[\frac{1}{2^{\ell+1}} \right] &\ge \epl \cdot p_i\left[\frac{1}{2^{\ell}} \right] \tag{L} \label{itm:L}
\\[0.8em]
p_i\left[\frac{1}{2^{\ell+1}} \right] &\le \epu \cdot p_i\left[\frac{1}{2^{\ell}} \right] \tag{U} \label{itm:U} 
\end{align}

These constraints ensure the CDFs are neither too flat nor too steep, enabling sufficient competition across neighborhoods of agents in a round.  

The Mixed-Strength Agents Conditions~\eqref{itm:L} and~\eqref{itm:U} are motivated by biased bits:
\begin{lemma}
The Mixed-Strength Agents Conditions~\eqref{itm:L} and~\eqref{itm:U} are satisfied by CDF $p_i[x]$ for agent $i$ that uses bit bias $q_i \in [\epl,\epu]$ to choose its random number in a round of the protocol.
\end{lemma}
\begin{proof}
    Since $q_i \ge \epl$, 
    \[ p_i\left[\frac{1}{2^{\ell+1}} \right] = q_i^{\ell+1}\ge \epl \cdot
     q_i^{\ell} = \epl \cdot p_i\left[\frac{1}{2^{\ell}} \right], \]
     and thus Condition~\eqref{itm:L} is satisfied.  Similarly, since
     $q_i \le \epu$, 
    \[ p_i\left[\frac{1}{2^{\ell+1}} \right] = q_i^{\ell+1}\le \epu \cdot
     q_i^{\ell} = \epu \cdot p_i\left[\frac{1}{2^{\ell}} \right], \]
     and thus Condition~\eqref{itm:U} is satisfied.
\end{proof}
Thus, agents that use fair bits or biased bits to choose their random number in each round of the protocol satisfy the Mixed
Strength Agents Conditions.  However, the Mixed-Strength Agents Conditions are satisfied
by a broader family of CDFs.  
As a simple example, each
bit that agent $i$ uses to choose its random number in a round can be chosen with a different bias.  

\section{Results and Analysis}
\label{analysis}

In this section, we analyze the performance of the MIS protocol under the mixed-strength agents model, assuming the agent-specific CDFs satisfy the technical conditions introduced in Section~\ref{sec:model}. 

We begin with a technical lemma that is used in the proof of
Lemma~\ref{lemma:main}.
\begin{lemma}
\label{lemma:upper}
Suppose the CDF $p_i[x]$ for agent $i$ 
satisfies the Mixed-Strength Agents Conditions.  Then,
\[ p_i\left[\frac{1}{2^{\ell}} \right]  \le 
2^{\ell \cdot \log\left({\epu}\right)} \] 
\end{lemma}
\begin{proof}
    By induction.  It can be checked by inspection to be true
    for $\ell=0$.  Suppose it is true for $\ell$.  Then, 
    \[p_i\left[\frac{1}{2^{\ell+1}}\right] \le 
    \epu \cdot p_i\left[\frac{1}{2^{\ell}}\right] \le
    \epu \cdot 2^{\ell \cdot \log\left({\epu}\right)} 
    = 2^{\log({\epu})} \cdot 2^{\ell \cdot \log\left({\epu}\right)}  
    = 2^{(\ell+1) \cdot \log\left({\epu}\right)}, \] where the first inequality is from Mixed-Strength Agents Condition~\eqref{itm:U} and the second inequality is from the induction hypothesis.
\end{proof}

The main technical result of this paper is the following.

\begin{lemma}
[Round Complexity of MIS in the Mixed-Strength Agents Model]
\label{lemma:main}
Let $G = (V, E)$ be a graph with $n$ agents and maximum degree $d$. Suppose each agent $i \in V$ chooses its
random number in each round with respect to a CDF $p_i[x]$ that satisfies the Mixed-Strength Agents Conditions~\eqref{itm:L} and~\eqref{itm:U}. Then the MIS protocol described in \cite{luby1986} completes in at most
\begin{equation}
\label{eq 0}
 \frac{8\cdot \log(n) \cdot \log\left(2\cdot d/\epl \right)}{\epl \cdot \log\left({1/\epu}\right)}
\end{equation}
rounds on average.
\end{lemma}

\begin{proof}
Let $N(i)$ denote the neighborhood of agent $i$, including $i$ itself, at the start of some round. Define the total CDF over the neighborhood of agent $i$ as
\[
t_i[x] = \sum_{j \in N(i)} p_j[x].
\]

Let $\ell_i$, hereafter called the {\bf level of agent $i$}, be the largest integer $\ell$ for which the total CDF over the neighborhood of agent $i$ up through $1/2^\ell$ at the start of the round exceeds a threshold:
\begin{equation}
\label{ge eq}
\ell_i = \max \left\{ \ell : t_i\left[ \frac{1}{2^\ell} \right] \ge \frac{\epl}{2}\right\}.
\end{equation}

Define:
\begin{itemize}
    \item $\mathcal{A}_\ell = \{ i : \ell_i = \ell \}$, the set of agents at level $\ell$,
    \item $\ell_{\max} = \max_i \ell_i$, the highest level among all agents,
    \item $x_{\max} = \frac{1}{2^{\ell_{\max}}}$.
\end{itemize}

We first show that
\begin{equation}
\label{eq 2}
\ell_{\max} \le \frac{\log\left(2 \cdot d/\epl\right)}{\log\left({1/\epu}\right)}.
\end{equation}
To see this, fix any agent $i \in \mathcal{A}_{\ell_{\max}}$. By Lemma~\eqref{lemma:upper}, for each $j \in N(i)$ we have 
\[ p_j[x_{\max}] \le x_{\max}^{\log\left({1/\epu}\right)} \] so
\[
t_i[x_{\max}] \le d \cdot \left( \frac{1}{2^{\ell_{\max}}} \right)^{\log({1/\epu})}.
\]
But by Equation~\eqref{ge eq}, \( t_i[x_{\max}] \ge \frac{\epl}{2} \). Rearranging gives Equation~\eqref{eq 2}.

Next, we show that \( \ell_{\max} \) decreases by at least one every \( O(\log(n)) \) rounds on average.
Fix a round. By the definition of \( \ell_{\max} \), for all agents \( i \), we have:
\[
t_i\left[ \frac{1}{2^{\ell_{\max} + 1}} \right] < \frac{\epl}{2}.
\]
Using Mixed-Strength Agents Condition~\eqref{itm:L}, this implies:
\begin{equation}
\label{eq 2.5}
t_i[x_{\max}] < \frac{1}{2}.
\end{equation}

Fix an agent \( i \in \mathcal{A}_{\ell_{\max}} \). We claim that agent $i$ is eliminated with probability at least \( \epl/8 \) in this round. Specifically, agent $i$ is eliminated if some neighbor \( j \in N(i) \) chooses the minimum value in \( N(j) \cup N(i) \). The probability of this is at least:
\begin{equation}
\label{eq 3}
\sum_{j \in N(i)} p_j[x_{\max}] \cdot \prod_{k \in N(j) \cup N(i) \setminus \{j\}} (1 - p_k[x_{\max}]).
\end{equation}

By Equation~\eqref{eq 2.5}, for each \( j \in V \),
\begin{equation}
\label{eq:aa}
\sum_{k \in N(j)} p_k[x_{\max}] < \frac{1}{2}
\end{equation}
Since for any sequence of positive numbers $a_0,a_1,\ldots,a_n$
such that $\sum_{i=0}^{n} a_i < 1/2$ it is the case that
\[ \sum_{i=0}^{n} (1- a_i) \ge 
1- \sum_{i=0}^{n} a_i > \frac{1}{2}.\]
This and Equation~\eqref{eq:aa} implies that, for each \( j \in V \),
\[
\prod_{k \in N(j)} (1 - p_k[x_{\max}]) > \frac{1}{2}.
\]
From this it follows that for any $j \in N(i)$, 
\begin{equation}
\label{eq 4}
\prod_{k \in N(j) \cup N(i)  \setminus \{j\}} (1 - p_k[x_{\max}]) \ge \prod_{k \in N(j)} (1 - p_k[x_{\max}])
\cdot \prod_{k \in N(i)} (1 - p_k[x_{\max}]) > \frac{1}{4}.
\end{equation}
Since \( i \in \mathcal{A}_{\ell_{\max}} \), i.e., 
\( \ell_i = \ell_{\max}\), from Equation~\eqref{ge eq} it follows that: 
\begin{equation}
\label{eq:3p}
  \sum_{j \in N(i)} p_j[x_{\max}] = t_i[x_{\max}] \ge \frac{\epl}{2}  
\end{equation} 
Combining Equations~\eqref{eq 3}, \eqref{eq:3p}, and \eqref{eq 4}, we conclude that the probability agent \( i \) is eliminated is at least:
\[
\frac{\epl}{2} \cdot \frac{1}{4} = \frac{\epl}{8}.
\]

Thus, since there are at most $n$ agents in 
\( \mathcal{A}_{\ell_{\max}} \), all such agents are 
eliminated in expectation after at most 
\[ \frac{8 \cdot \log(n)}{\epl} \] 
rounds and thus 
$\ell_{\max}$ decreases at least one on average each \( \frac{8 \cdot \log(n)}{\epl}  \) rounds. 
Since \( \ell_{\max} \) starts at most at
\[\frac{\log\left(2 \cdot d/\epl\right)}{\log({1/\epu})}  \] 
the total expected number of rounds is bounded by Equation~\eqref{eq 0}, as claimed.
\end{proof}
As a concrete example, if $\epl = 1/4$ and $\epu = 1/2$ 
then the bound in Lemma~\ref{lemma:main} simplifies to
\begin{equation}
32\cdot \log(n) \cdot \log\left(8\cdot d \right).
\end{equation}

Corollary~\ref{cor:main} frames Lemma~\ref{lemma:main} in the context of the gap between $0$ and 
$\epl$ and the gap
between $\epu$ and $1$. 
\begin{corollary}
\label{cor:main}
Let $G = (V, E)$ be a graph with $n$ agents and maximum degree $d$. Suppose each agent $i \in V$ chooses its
random number in each round with respect to a CDF $p_i[x]$ that satisfies the Mixed-Strength Agents Conditions~\eqref{itm:L} and~\eqref{itm:U}. Then the MIS protocol described in \cite{luby1986} completes in at most
\begin{equation}
O\left(\frac{\log(n) \cdot \log(d)}{\epl \cdot \epub}\right)
\end{equation}
rounds on average, where $\epub = 1 - \epu$.
\end{corollary}
\begin{proof}
    As $\epu$ approaches 1, the term $\frac{1}{\log(1/\epu)}$ approaches $\frac{1}{1-\epu} = 1/\epub$.
\end{proof}

The main result in this paper, Theorem~\ref{thm:main}, follows from Lemma~\ref{lemma:main}. 
\begin{theorem}
\label{thm:main}
    If each agent’s CDF satisfies the Mixed-Strength Agents Conditions
    then the MIS protocol from \cite{luby1986} completes in
\[
O(\log(n) \cdot \log(d))
\]
rounds with high probability, where $n$ is the number of agents and $d$ is the maximum degree in the graph. The hidden constant depends on the constants $\epl$ and $\epu$ of the Mixed-Strength Agents Conditions.
\end{theorem}

\section{Example Graph Where Few Edges Are Eliminated}
\label{sec:few edges}

In the classical MIS protocol of~\cite{luby1986}, each agent independently generates a random number by tossing fair coins, and joins the MIS if its number is strictly smaller than those of all its neighbors. The analysis shows that, in expectation, a constant fraction of the edges is eliminated in each round.

In contrast, we show that when agents use \emph{biased} coins, there are graphs where only a \emph{small} fraction of edges are eliminated per round on average. Our construction shows how a hierarchy of coin biases can significantly slow down the progress of the MIS protocol.

\vspace{0.5em}
\noindent
\textbf{Graph construction.}
The graph consists of \( \sqrt{\log(n)} + 1 \) cliques, each with \( n \) agents. These cliques are arranged in a line, with a complete bipartite graph connecting each pair of consecutive cliques.

Each agent generates a random number by flipping biased coins. The bit bias varies across cliques:
\begin{itemize}
    \item Agents in the first clique use fair bits, i.e., $q_i = 1/2$.
    \item Agents in the $j^{\rm th}$ clique use bit bias
    \[
    q_i = \frac{1}{2} - \frac{j-1}{4 \cdot \sqrt{\log(n)}},
    \]
    so agents become gradually weaker from left to right. In the final clique, $q_i = 1/4$.
\end{itemize}

\vspace{0.5em}
\noindent
\textbf{Effect on MIS selection.}
Because agents in earlier cliques are more strongly biased toward 0, they are more likely to generate the smallest numbers. For any pair of consecutive cliques, the agent with the lowest number is very likely to come from the earlier clique, making it unlikely that any agent in the later clique joins the MIS.

In the first round, with high probability, a single agent from the first clique joins the MIS. This agent is adjacent to all agents in both the first and second cliques, so both cliques become inactive. In the second round, the same pattern repeats with the third and fourth cliques, and so on.

This continues for about \( \frac{\sqrt{\log(n)}+1}{2} \) rounds. Since each round eliminates only two cliques out of \( \sqrt{\log(n)}+1 \), the fraction of edges removed per round is at most
\[
\frac{2}{\sqrt{\log(n)} + 1}.
\]
This is captured in the following lemma.

\begin{lemma}
\label{lem:few-edges}
There is a graph and choice of coin biases under the asymmetric model that satisfy the Mixed-Strength Agents Conditions such that in each round of the MIS protocol, only a 
\[
O\left( \frac{1}{\sqrt{\log(n)}} \right)
\]
fraction of the total edges are eliminated on average.
\end{lemma}

\begin{proof} \emph{(Sketch)}
Let $\ell = \log(n)$. Consider the first round. With constant probability, some agent in the first clique generates a number whose first $\ell$ bits are all 0. In the second clique, where each bit is 0 with probability \( \frac{1}{2} - \frac{1}{4 \cdot \sqrt{\log(n)}} \), the chance that \emph{any} agent produces a number with $\ell$ leading zeros is at most
\[
n \cdot \left( \frac{1}{2} - \frac{1}{4 \cdot \sqrt{\log(n)}} \right)^\ell = 2^{-\Omega(\sqrt{\log(n)})}.
\]
From this it can be shown that with high probability an agent from the first clique is added to the MIS and no agent from the second clique is added to the MIS, and thus both cliques become inactive.

A similar argument applies to all other pairs of consecutive cliques. In each such pair, the earliest clique likely contains some agent whose value is smaller than that of every agent in the next clique. As a result, all cliques from the third onward remain active at the end of the first round.

The same argument holds for the second and subsequent rounds: one node from the first remaining clique is added to the MIS, the first two remaining cliques become inactive, and all nodes in the remaining cliques stay active.

Since each round eliminates only two cliques, and there are \( \Theta(\sqrt{\log(n)}) \) cliques overall, the number of rounds is \( \Theta(\sqrt{\log(n)}) \), and the average fraction of edges removed per round is \( O(1/\sqrt{\log(n)}) \).
\end{proof}

\vspace{0.5em}
\noindent
\textbf{Sharper bound.}
With a more carefully chosen sequence of biases, this bound can be improved to
\[
O \left( \frac{1}{\log(n)} \right)
\]
edges eliminated per round on average.

\section{Discussion and Open Questions}
\label{sec:discussion}

This work introduces and analyzes the mixed-strength (asymmetric) agents model for distributed MIS selection, capturing localized competition among agents of varying capabilities. We showed that the classical MIS protocol from \cite{luby1986} continues to function under this broader model, completing in $O(\log(n) \cdot \log(d))$ rounds with high probability, even when agents draw their strength values from different distributions. We also demonstrated that, unlike in the symmetric case, the presence of asymmetric CDFs can lead to significantly slower progress in certain graphs, with only a 
$O\left(\frac{1}{\sqrt{\log (n)}}\right)$ fraction of edges eliminated in each round.

One open question is whether the bounds in Lemma~\ref{lemma:main} and Theorem~\ref{thm:main} are asymptotically tight, or whether it is possible to prove a tighter $O(\log(n))$-round bound for the MIS protocol in the mixed-strength agents model.

While the model captures heterogeneity in agent strength, it is fundamentally static: once agents generate a maximal independent set, the system stabilizes and no further reconfiguration occurs. Furthermore, each agent's CDF remains fixed over time. However, in many natural and engineered settings, agent capabilities evolve. Agents may weaken or strengthen, enter or leave the system, or adapt their behavior in response to external conditions. This raises the following question:

\begin{quote}
\emph{Can the MIS protocol be extended to respond dynamically to changes in agent strength, presence, or connectivity, and if so, how efficiently?}
\end{quote}

This question lies at the intersection of local adaptability and global consistency, and may require new algorithmic primitives beyond the classical MIS framework. Preliminary work on dynamic MIS addresses evolving graphs and population turnover. Extending the mixed-strength framework to support such dynamic environments remains an open and intriguing direction for future research.

A further observation arises from the construction in Section~\ref{sec:few edges}. There, MIS selection progresses sequentially through the cliques, even though the process is randomized. This emergent directionality is a consequence of the structured asymmetry in agent CDFs across the graph. It suggests the possibility that such asymmetries could be leveraged to support simple forms of local coordination or distributed decision-making. Formalizing and understanding these possibilities is another avenue for future exploration.

\bibliography{main}

\end{document}